\documentclass[10pt,twocolumn]{IEEEtran}
\makeatletter
\def\ps@headings{%
\def\@oddhead{\mbox{}\scriptsize\rightmark \hfil \thepage}%
\def\@evenhead{\scriptsize\thepage \hfil \leftmark\mbox{}}%
\def\@oddfoot{}%
\def\@evenfoot{}}
\makeatother
\usepackage[dvips]{graphicx}
\usepackage{amssymb}
\usepackage{mathrsfs}

\def\bfx{{\bf x}}

\def\bfone{{\bf 1}}
\newcommand{\lengthx}[1]{\vert #1 \vert}

\newtheorem{result}{\bf{Result}}
\newtheorem{lemma}{\bf{Lemma}}
\newtheorem{theorem}{\bf{Theorem}}

\title{Analysis of Geometric Disaster Evaluation Model for Physical Networks}
\author{Hiroshi Saito\\
NTT Network Technology Laboratories\\
 3-9-11, Midori-cho, Musashino-shi, Tokyo 180-8585, Japan\\
E-mail: saito.hiroshi@lab.ntt.co.jp, URL: http://www9.plala.or.jp/hslab/\\
Phone: +81 422 59 4300, Fax: +81 422 59 5671.}

\date{}

\begin{document}

\maketitle

\begin{abstract}
A geometric model of a physical network affected by a disaster is proposed and analyzed using integral geometry (geometric probability). This analysis provides a theoretical method of evaluating performance metrics, such as the probability of maintaining connectivity, and a network design rule that can make the network robust against disasters.

The proposed model is of when the disaster area is much larger than the part of the network in which we are interested.
Performance metrics, such as the probability of maintaining connectivity, are explicitly given by linear functions of the perimeter length of convex hulls determined by physical routes.
The derived network design rule includes the following.
(1) Reducing the convex hull of the physical route reduces the expected number of nodes that cannot connect to the destination. 
(2) The probability of maintaining the connectivity of two nodes on a loop cannot be changed by changing the physical route of that loop. 
(3) The effect of introducing a loop is identical to that of a single physical route implemented by the straight-line route.

\end{abstract}

\section{Introduction}
A massive earthquake occurred on March 11, 2011 in northeastern Japan.
The earthquake and the resulting tsunami destroyed everything including network facilities.
According to the press release from NTT \cite{ntt}, ^^ ^^ Facilities were damaged and commercial power supply was disrupted at exchange offices, among other things, impacting approximately 1.5 million circuits for fixed-line services, approximately 6,700 pieces of mobile base-station equipment, approximately 15,000 circuits for corporate data communication services, and others." 
It was also reported that transmission lines were disconnected in 90 routes, 18 exchange office buildings were destroyed, and 23 buildings were submerged, approximately 65,000 telephone poles were destroyed, and submersion and physical damage to aerial cables reached about 6,300 kilometers. 
Such severe damage and problems in a network cause secondary negative effects over a disaster area, for example, emergency calls not being able to go through.

One may assume that such a huge disaster cannot re-occur in the future.
Unfortunately, this is not the case.
A similar level of damage and destruction in telecommunications networks due to the earthquake on May 12, 2008 in China was reported \cite{china}.
Minor damage due to disasters can occur anywhere and at any moment.
Therefore, it is important to construct networks robust against disaster.
Higher layer functions are normally useless when the physical connectivity of fibers is lost or network node buildings are lost.
The physical design, including the geographical design of the network and its evaluation, plays a central role.

Through such experiences, network operators have made efforts to increase network robustness against disasters.
For example, a large earthquake struck in 1968 in northern Japan triggering the expansion of microwave transit systems as a means of increasing network survivability and the number of physical routes because the earthquake disconnected the network in Hokkaido, the second largest island in Japan, from the main island of Honshu due to the disconnection of the submarine cable between the two islands \cite{ntt-east}.
Another earthquake in 1993 resulted in the development of a transportable earth station of a satellite communication system \cite{ntt-east}.
Nevertheless, we do not have mathematical models or mathematic frameworks for handling the effects of disasters striking in an unknown geographical location or area.
This paper responds to the needs of such models or frameworks.

There have been a large number of theoretical papers published evaluating the reliability, availability, and survivability for a given network.
However, most of these papers focus on a single failure (or independent failures) of a network node or a link.
They take into account the topology of the network rather than its physical shape.
For a given set of network topologies and failure rates of network entities, evaluation of a metric, such as the probability that a pair of network entities can be connected, is a typical example \cite{book}.
In a disaster, the assumption of a single or independent failure is not valid.
In addition, the physical shape of a network is important for evaluating the impact of a disaster on network survivability; however, most studies have not covered this point.

Several studies of network survivability by taking into account correlated failure and geometric/geographical conditions have been reported.
Grubesic \cite{geography} evaluated the network survivability of the current Internet based on geographical data.
Although he focused on the physical route of a network, it was a case study, and no mathematical models or methods were provided.
Liew and Lu \cite{survivability} proposed a framework to evaluate network survivability during a disaster and introduced a survivability function to various metrics.
Although their framework can introduce correlated failures, they did not propose any method or model of correlations, and they assumed that the independence of failures and the failure rate was proportional to the link length.
Furthermore, they did not consider the physical shape of the disaster area or that of the network due to the lack of a mathematical framework.
Wu et al. \cite{underseaCableFailure} discussed the optimization of the physical route of an undersea cable by assuming a disk-shaped disaster area.
By assuming a rectangular route, the length of an edge is determined by minimizing cost while maintaining a higher probability of connecting two cities than the threshold.

Taking into account a disaster area, the minimum number of cuts disconnecting the source and sink nodes was discussed in the following papers.
As far as we know, Bienstock \cite{OR} initiated the study of this problem.
Algorithms computing the minimum number of disaster areas disconnecting the source and sink nodes were investigated when all the edges intersecting the disaster areas are removed.
Sen et al. \cite{sen} proposed a region-based connectivity as a metric for fault-tolerance.
Assuming the region is a disk-shaped disaster area, polynomial time algorithms calculating region-based connectivity are provided.
Neumayer et al. \cite{discMinCut} discussed the geographical min-cut, defined as the minimum number of disk-shaped disaster areas to disconnect a pair of nodes, and the geographical max-flow, defined as the maximum number of paths that are not disconnected by a single disaster area, and showed that geographical min-cut is not equal to geographical max-flow.
Agarwal et al. studied algorithms that find a disaster location having the highest expected impact on a network, where the impact is defined by various metrics such as the number of failed components \cite{wdmFailure}.

Recently, Neumayer et al. published two papers intended to cover network survivability in a disaster \cite{failureToN}, \cite{failureINFOCOM}.
In one of their introductions \cite{failureToN}, they stated that, ^^ ^^ we are the first to attempt to study this problem."
In their network model, there is a set of line segments of which end points are locations of network center buildings and the disaster model is a line segment or a circle \cite{failureToN}.
They proposed to use an optimization technique to find the worst case disaster.
On the other hand, Neumayer and Modiano \cite{failureINFOCOM} used geometric probability (integral geometry) to model the randomness of a disaster.
Their network model is, again, a set of line segments of which end points are locations of network center buildings and the disaster model is a line.
These papers also emphasize a polynomial order algorithm to evaluate metrics.

Forecasting the time, even the location and size of a disaster, is often difficult.
For example, although research on earthquakes in Japan is quite extensive, the size of a disaster area caused by the earthquake on March 11, 2011 was far beyond expectations.
Towns that were expected to be free from tsunami were engulfed and districts that had never experienced a strong earthquake were affected by intense shaking.
This shows that the frequency of disasters is too low to make an accurate disaster area model based on statistics. There are also many types of disasters, such as earthquakes, tsunami, hurricanes, landslides, wild fires, tornados, extraordinarily heavy snowfall, and electromagnetic pulses. Furthermore, network operators have many subnetworks such as subscriber and regional networks.
There are too many combinations of disaster and subnetworks to evaluate all cases individually.
Therefore, it is beneficial for network operators to use a method for assessing network survivability applicable to generic subnetworks without specific disaster information and to provide a rule of thumb for designing robust networks. 
Based on this motivation, a disaster area is modeled as an area randomly placed around a network.
With extensive use of integral geometry (geometric probability), more generic network models than those used by Neumayer and Modiano \cite{failureINFOCOM} can be covered and explicit formulas for performance metrics, such as the probability that a pair of nodes still maintain connections, can be derived.
The analyzed results lead us to a network design that can improve robustness against a disaster.

Integral geometry (geometric probability) used in this paper is a mathematical method for evaluating the measures that a certain set (normally a subset of a plane) satisfies certain characteristics and has been used in several papers regarding sensor networks as well as that by Neumayer and Modiano \cite{failureINFOCOM}.
For example, a series of papers \cite{infocom}, \cite{mobileComp}, \cite{signal}, based on analysis using integral geometry, proposed shape estimation methods for a target object based on reports from sensor nodes of unknown locations.
Lazos et al. \cite{detection} and Lazos and Poovendran \cite{lazos} directly applied the results of the integral geometry discussed in Chapter 5 Section 6.7 of \cite{Santalo} to the analysis of detecting an object moving in a straight line and to the evaluation of the probability of $k$-coverage. Kwon and Shroff \cite{routing} also applied integral geometry to the analysis of straight-line routing, which is an approximation of the shortest path routing, and Choi and Das \cite{energy} used it to determine sensors in energy-conserving data gathering.

The contributions of this paper are as follows. (1) This paper discusses a theoretical method for evaluating network survivability metrics, such as connectivity probability, when the physical network and disaster area shapes are given. The information on the location of the disaster area is not needed. (2) This theoretical method explicitly reveals the relationships between network survivability metrics and physical network shape, when the disaster area is geographically much larger than the length between two nodes of interest and its boundary is macroscopically a line. For example, the probability of connectivity between two nodes connected by a single route is a linear function of the perimeter length of the convex hull of the physical route. (3) By using the relationships between network survivability metrics and physical network shape, principles useful in designing a physical route of a network are derived. For example, (i) reducing the convex hull of the physical route reduces the expected number of nodes disconnected to the destination; (ii) the probability of maintaining the connectivity of two nodes on a ring network is insensitive to the physical route shape of that network; (iii) although the straight-line route maximizes the connection probability when a single route is used, the effect of introducing a loop is identical to that when the straight-line route is not practical. As a result, we can obtain a physical network shape that is robust against a disaster and the rule of thumb for designing such a physical network. (4) The proposed simple model, as disaster area shape, is valid when the bumps of its boundary are small or the interval of two bumps are short compared with the mean distance between two nodes.

The organization of this paper is as follows. Section II introduces a basic explanation of integral geometry and geometric probability and derives a basic theorem for the analysis in later sections.
Section III explains the proposed model.
Section IV discusses the analysis results when the disaster area is much larger than the area of interest, and Section V discusses numerical examples.
Section VI concludes the paper. 

\section{Preliminaries}
To analyze the relationship between the shape of the network and the effect of a disaster, this section provides mathematical preliminaries.

\subsection{Introduction of integral geometry and geometric probability} \label{integral-geometry}
The concept of integral geometry and geometric probability \cite{Santalo} is introduced as a preliminary for evaluating a disaster occurring at a random location. 


For a bounded set $K$ in 2-dimensional space $\mathbb{R}^2$, we can define the motion-invariant measure of the set of positions of $K$ satisfying the condition $X$, where the position of $K$ is determined by the position of its reference point $(x,y)$ and by its direction $\theta$ formed by a reference line fixed in $K$ with another reference line fixed to the fixed coordinates.
An example of $X$ is $K\cap K_0\neq \emptyset$ for a given $K_0$.
By moving $K$ over a domain satisfying $X$ in the parameter space of $(x,y,\theta)$, we obtain the motion-invariant measure $m(K;X)=\int_X dK=\int_X dx\, dy\,  d\theta$ of the set of positions of $K$ satisfying $X$.
That is, $m(K;X)$ is the area size of the parameter space of $(x,y,\theta)$ satisfying $X$.
This ^^ ^^ moving $K$" is denoted as $dK$ and called the kinematic density.
In fact, $dK=dx\,dy\,d\theta$.  That is, ^^ ^^ moving $K$" is ^^ ^^ moving $(x,y,\theta)$".
Consequently, the measure $m(K;X)$ is defined by the size of the area in which the positions of $K$ satisfy $X$ in the parameter space of $(x,y,\theta)$.

Once we have defined the measures $m(K;X)$ and $m(K;Y)$ ($Y\subseteqq X$), the probability that the position of $K$ satisfying $Y$ among the position of $K$ satisfying $X$ is defined by the quotient of measures $m(K;Y)/m(K;X)$.
This is called a geometric probability.
That is, the geometric probability is proportional to the measure $m(K;Y)$ and is normalized by $m(K;X)$.
In this sense, the measure $m(K;Y)$ is a non-normalized probability when the parameters $(x,y)$ and $\theta$ of $K$ move uniformly over the parameter space.

Similarly, we can define the measure of a set of lines.
Consider a line $G$ determined by the angle $\theta$, in which the direction perpendicular to $G$ makes a fixed angle with the positive part of the $x$-axis ($-\pi \leq \theta \leq \pi$), and by its distance $p$ from the origin $O$ ($0\leq p$) (Fig. \ref{lineset}). That is, $G$ is specified by the coordinates $(p,\theta)$. 
By using the coordinates $(p,\theta)$ and the criterion of motion-invariance, the motion-invariant measure of the set of lines $G(p,\theta)$ satisfying $X$ is defined by the simple integral form $m(G;X)=\int_X dp\,d\theta$ \cite{Santalo}.
(We can use another parameterization, but we cannot use this simple integral form.
This is because integral geometry requires the calculated results to be invariant under the group of motions in the plane and because another parameterization requires a complicated form to make $m(G;X)$ motion-invariant.)

\begin{figure}[htb] 
\begin{center} 
\includegraphics[width=6cm,clip]{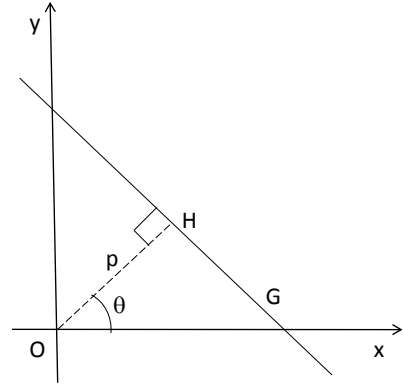} 
\caption{Parameterization of line} 
\label{lineset} 
\end{center} 
\end{figure}

When we consider a strip $B$ with breadth $w$, its position should be parameterized by the position of its midparallel line and its direction.
(The strip is the area between two parallel lines.)
Thus, $(p,\theta)$ of the midparallel line can be used for $B$.
Then, the integral form $m(B;X)=\int_X dp\,d\theta$ can also be used to evaluate the measure $m(B;X)$.
A randomly placed $B$ or a randomly placed $G$ takes the position determined by the parameters $(p,\theta)$ of $B$ or $G$, which uniformly take values over the parameter space.

Throughout this paper, any boundary of a set in $\mathbb{R}^2$ is smooth and differentiable except for the finite number of points.

\subsection{Lemma and theorem}
This subsection provides a lemma and a theorem for the analysis in the following sections.
For a bounded area $C$ in $\mathbb{R}^2$, a directional $G$, and a directional $B$ with $w$, we use the following notations in the remainder of this paper. 
\begin{itemize}
\item $\partial C$: boundary of $C$, 
\item $\lengthx{C}$: perimeter length of $C$, 
\item $\overline C$: convex hull of $C$,
\item $R_G$: right-half plane (the plane of the right side) of $G$,
\item $R_B$ : plane of the right side of and outside $B$,
\item $G_B$: left-side boundary line of $B$.
\end{itemize}
($R_G$ and $R_B$ are located at the right of line $G$ or strip $B$ when the direction of $G$ or $B$ is upward.)
Note that $G_B$ is the boundary of $B\cup R_B$.
That is, $R_{G_B}=(B\cup R_B)$.

In later sections, $C$ in the following lemma and theorem is, for example, a route between two nodes.
Then, the following theorem provides a probability that a route intersects a disaster area.

\begin{lemma}\label{lemma1}
Let $C$ be a set in $\mathbb{R}^2$.
The event $\{R_{G_B} \cap C= \emptyset\}$ is equivalent to $\{R_{G_B}  \cap \overline{C}= \emptyset\}$, and the event $\{R_G \cap C= \emptyset\}$ is equivalent to $\{R_G\cap \overline{C}= \emptyset\}$.
\end{lemma}

This is almost trivial because $R_{G_B}$ or $R_G$ is a half-plane.

\begin{theorem}
Let $C$ be a set in $\mathbb{R}^2$ and assume that $C\subseteqq A_0$ where $A_0$ (area of interest) is bounded and convex.

Then, 
\begin{equation}
\Pr(R_{G_B} \cap C= \emptyset\vert B\cap A_0\neq \emptyset)=\frac{\lengthx{A_0}-\lengthx{\overline{C}}}{2\lengthx{A_0}+2\pi w}.\label{pr-1}
\end{equation}
\end{theorem}
\begin{proof}
First, assume that $C$ is a convex set.
Take an origin in $C$ and consider support functions $p_0=p_0(\theta)>0$ and $q_0=q_0(\theta)>0$ for $C$ and $A_0$.

In general, a support function for $X$ provides a set of lines, and each line in the set is called a line of support for $X$.
A line of support for $X$ is a line containing at least one of $X$ but such that one of the two open half planes determined by the line contains no point of $X$ (Fig. \ref{support}) \cite{dictionary}.
The following relationship is known between a support function $p_x(\theta)$ for $X$ and its perimeter length $\lengthx{X}$: $\int_{-\pi}^\pi p_x(\theta) d\theta=\lengthx{X}$ \cite{Santalo}.
For each $\theta$, the distance of a line of support from the origin is $p_x(\theta)$.

\begin{figure}[htb] 
\begin{center} 
\includegraphics[width=8cm,clip]{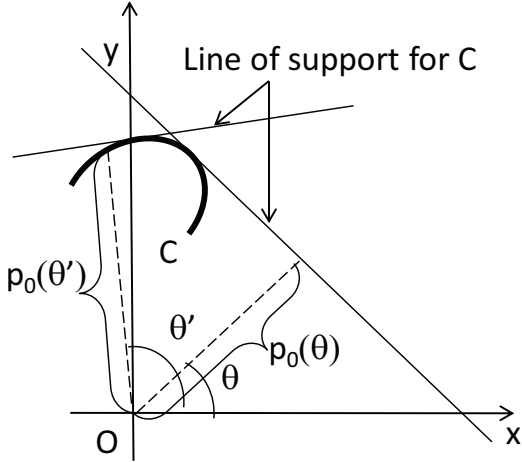} 
\caption{Support function $p_0=p_0(\theta)$} 
\label{support} 
\end{center} 
\end{figure} 

For a fixed $\theta$, the range of $p$ that satisfies $R_{G_B}  \cap C\neq \emptyset, B\cap A_0\neq \emptyset$ is $-q_0(\theta+\pi)-w/2\leq p\leq p_0(\theta)+w/2$  (See Fig. \ref{p_0}.
In this figure, the two strips correspond to the two positions of $B$.
One position corresponds to $p= p_0(\theta)+w/2$ and the other corresponds to $p=-q_0(\theta+\pi)-w/2$.
Because $B$ is directional, we consider $-\infty <p < \infty, -\pi\leq \theta<\pi$.)
Thus,
\begin{eqnarray}
&&m(B; R_{G_B} \cap C\neq \emptyset, B\cap A_0\neq \emptyset)\cr
&=&\int_{R_{G_B} \cap C\neq \emptyset, B\cap A_0\neq \emptyset} dp\   d\theta\cr
&=&\int_{-\pi}^\pi (p_0(\theta)+q_0(\theta+\pi)+w)d\theta.
\end{eqnarray}
Note that $p_0$ is a support function of $C$. Thus, $\int_{-\pi}^\pi p_0(\theta) d\theta=\lengthx{C}$.
Similarly, $\int_{-\pi}^\pi q_0(\theta+\pi) d\theta=\lengthx{A_0}$.
Therefore,
\begin{equation}
m(B; R_{G_B} \cap C\neq \emptyset, B\cap A_0\neq \emptyset)=\lengthx{C}+\lengthx{A_0}+2\pi w.\label{basic}
\end{equation}

\begin{figure}[htb] 
\begin{center} 
\includegraphics[width=8cm,clip]{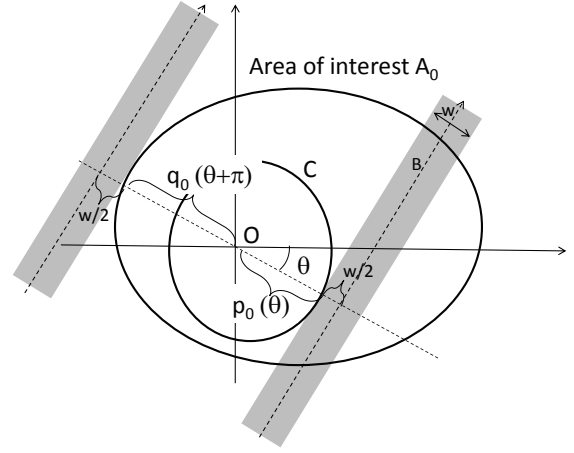} 
\caption{Derivation of $m(B; R_{G_B} \cap C\neq \emptyset, B\cap A_0\neq \emptyset)$} 
\label{p_0} 
\end{center} 
\end{figure}

Similarly,
\begin{eqnarray}
m(B; B\cap A_0\neq \emptyset)&=&\int_{B\cap A_0\neq \emptyset} dp \  d\theta\cr
&=&\int_{-\pi}^\pi q_0(\theta)+q_0(\theta+\pi)+w\, d\theta\cr
&=&2\lengthx{A_0}+2\pi w.
\end{eqnarray}
According to the definition of geometric probability,
\begin{eqnarray}
&&\Pr(R_{G_B} \cap C= \emptyset\vert B\cap A_0\neq \emptyset)\cr
&=&1-\Pr(R_{G_B}  \cap C\neq \emptyset\vert B\cap A_0\neq \emptyset)\cr
&=&1-\frac{m(B; R_{G_B}\cap C\neq \emptyset, B\cap A_0\neq \emptyset)}{m(B; B\cap A_0\neq \emptyset)}.
\end{eqnarray}
Thus, we obtain Eq. (\ref{pr-1}).

When $C$ is not a convex set, consider $\overline{C}$ instead of $C$.
Then, $\Pr(R_{G_B}  \cap \overline{C}= \emptyset\vert B\cap A_0\neq \emptyset)=\frac{\lengthx{A_0}-\lengthx{\overline{C}}}{2\lengthx{A_0}+2\pi w}$.
By using Lemma \ref{lemma1}, we obtain Eq. (\ref{pr-1}).
\end{proof}

\section{Model}
This paper investigates a physical network, such as an optical fiber network, within a bounded and convex $A_0$.
This $A_0$ is the area of interest, and disasters causing damage in part of $A_0$ are taken into account.

Let $n_i$ ($i=1,2,\cdots$) be the nodes in the network, and let $l(i,j)$ be the physical link between two consecutive nodes $n_i$ and $n_j$.
When there are $N(i,j)$ routes between nodes $n_i$ and $n_j$, let $s_k(i,j)$ be its $k$-th physical route.
The route $s_k(i,j)$ is a concatenation of physical links. For example, $s_k(i,j)$ consists of $l(i,k_1), l(k_1,k_2), l(k_2,j)$ when there are intermediate nodes $n_{k_1},n_{k_2}$ on the route between $n_i$ and $n_j$.
When there are no intermediate nodes between $n_i$ and $n_j$, $s_k(i,j)=l(i,j)$.
Path configuration $p(i,j)$ is defined as a union of the routes between $n_i$ and $n_j$.
That is, $p(i,j)\equiv \cup_{k=1}^{N(i,j)}s_k(i,j)$.

In this paper, the meaning of $s_k(i,j)$ and $p(i,j)$ is not limited to the connectivity between $n_i$ and $n_j$.
The meaning of $s_k(i,j)$ and $p(i,j)$ implies the physical route shape and its union between $n_i$ and $n_j$.
Therefore, we can define, for example, $\overline{p(i,j)}$, which is the convex hull of $p(i,j)$.

When $p(i,j)$ consists of a single route, that is, when $N(i,j)=1$ and $p(i,j)=s_1(i,j)$, we may call $p(i,j)$ a single route path.
When we call $p(i,j)$ a ring-type network, this means that there are two non-overlapping routes between $n_i$ and $n_j$.
For a ring-type network, it is assumed that there is connectivity between two nodes if at least a clockwise or counterclockwise route between these two nodes is maintained.
For a ring-type network, it is assumed that the area of which boundary is $p(i,j)$ is convex for the remainder of this paper.
(For simplicity, we denote ^^ ^^ $p(i,j)$ is convex" in the remainder of this paper. However, it formally means that the area of which boundary is $p(i,j)$ is convex.)

A network affected by a disaster is analyzed.
With no prior information of the disaster, the disaster area $D$ is modeled as a randomly placed area around a network in $\mathbb{R}^2$.
The disaster area $D$ is modeled as a realization of a spatially stationary process. 

It is assumed that the portion of the network included in $D$ does not work at all.
That is, no network elements function in $D$.

In the remainder of this paper, it is assumed that a disaster area $D\subset \mathbb{R}^2$ is geographically much larger than the length between two nodes of interest.
For example, the disaster area of a large earthquake is at least hundreds of km$^2$. Some may reach tens of thousands of km$^2$.
A large hurricane can create a disaster area larger than a hundred km$^2$.
Therefore, this assumption is useful, for example, for evaluating a disaster affecting a small subnetwork, such as a subscriber network, or for designing a robust physical route of such a network against disasters.
Because $D$ is very large, we can assume that its boundary area is macroscopically a line but microscopically has bumps.
The proposed model is a strip $B$ with a half-plane $R_B$.
That is, $B$ is a model in which the boundary is rugged within breadth $w$ (Fig. \ref{strip}).

For a directional $B$, $D$ is assumed to be $R_{G_B}$, that is $D\approx R_{G_B}$.
Because we assume that a part of the network in $R_{G_B}$ does not work at all, this model is identical to that in which the boundary of $D$ is a line $G_B$.
This assumption underestimates the probabilities of connectivities under the assumption that a part of the network in $D$ does not work at all.
However, as a network operator, this safe-side assumption is adopted to obtain a feasible rule of thumb.
To show the validity of this model of $D$, we discuss the simulation results in Section \ref{large-n}.

\begin{figure}[htb] 
\begin{center} 
\includegraphics[width=8cm,clip]{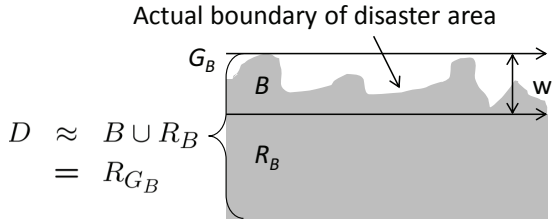} 
\caption{Strip model} 
\label{strip} 
\end{center} 
\end{figure}

The following is a list of notations used in the remainder of this paper.
\begin{itemize}
\item $i\leftrightarrow j$ means that there is connectivity between nodes $n_i$ and $n_j$,
\item $i \not \leftrightarrow j$ means that the connectivity is lost between $n_i$ and $n_j$, and
\item $\overline{(n_i,n_j)}$ means a straight-line segment of which the two end points are $n_i$ and $n_j$ (this is equivalent to the convex hull of $n_i$ and $n_j$).
\end{itemize}
Sometimes, for example, $i\leftrightarrow k\leftrightarrow j$ is used.
This means that $i\leftrightarrow k$ and $k\leftrightarrow j$.

\section{Analysis}\label{large}
In the rest of this section, cases in which $B\cap A_0\neq \emptyset$ are of focus.
To simplify the notation, the description for this area of focus is omitted in the remainder of the paper.

\noindent {\bf [Remark 1]}
This area of focus means that disasters causing damage in part of $A_0$ are taken into account but disasters that affect the entire $A_0$ are removed.
If those disasters affecting the entire $A_0$ are taken into account, the measure of $\{A_0\cap R_{G_B}\neq \emptyset\}$ is required.
However, this measure becomes infinite because $R_{G_B}$ is a half plane.
As a result, $\Pr (i\leftrightarrow j)$ cannot be appropriately determined.
Because of this technical reason, the cases in which $B\cap A_0\neq \emptyset$ are of focus.
However, even under this focus, the cases in which a subnetwork of interest is completely included in $R_{G_B}$ can be considered.
By appropriately setting $A_0$, there are no practical problems caused by this focus.

\noindent {\bf [Remark 2]}
The results in this paper are also valid when $D$ is a strip $B'$ with breadth $W$ ($W\gg w, W>d_{max}$) and $D$ satisfying $B' \cap A_0\neq \emptyset$ is the disaster we take into account (Fig. \ref{wide-strip}). Here $d_{max}$ is the maximum distance between two points in a physical path configuration. In this case, the numerator of all the results in this section are replaced with $m(B'; B'\cap A_0\neq \emptyset)=2\lengthx{A_0}+2\pi W$. If we believe that removing disasters that affect the entire $A_0$ is unnatural, this wide-strip model can be used because we can take into account disasters causing damage in part of $A_0$ as well as those damaging the entire $A_0$. However, the wide-strip model makes the proofs of the results less comprehensive than with the original model.

This validity, even under the wide strip model, is due to the equivalence of two events. The first one is that a part of the path configuration is in $R_{G_{B'}}$ and the other one is that this part is in $B'$. This equivalence is proven as follows. Because $B' \subseteqq R_{G_{B'}}$, if a part of the path configuration is in $B'$, then it is in $R_{G_{B'}}$. If a part $\Phi$ of the path configuration is in $R_{G_{B'}}$ and a point $\bfx\in \Phi$ is outside $B'$, the distance between $\bfx$ and a point outside $\Phi$ on this path configuration is larger than $W$. This contradicts the assumption that $d_{max}<W$. Therefore, if a part of the path configuration is in $R_{G_{B'}}$, then it is included in $B'$.

\begin{figure}[htb] 
\begin{center} 
\includegraphics[width=8cm,clip]{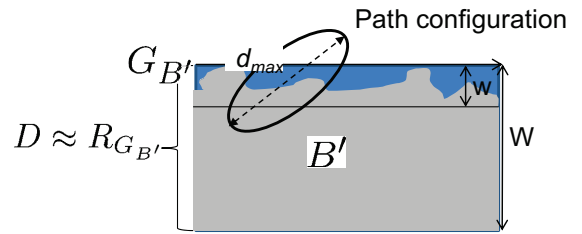} 
\caption{Wide strip model} 
\label{wide-strip} 
\end{center} 
\end{figure}

\subsection{Single route path}
\begin{result}\label{single-route-path}
When $p(i,j)$ consists of a single route, the probability $\Pr(i\leftrightarrow j)$ that there is connectivity between $n_i$ and $n_j$ is given by  
\begin{equation}
\Pr (i\leftrightarrow j)=\frac{\lengthx{A_0}-\lengthx{\overline{p(i,j)}}}{2\lengthx{A_0}+2\pi w}.\label{sl}
\end{equation}
\end{result}

\begin{proof}
The event $i\leftrightarrow j$ is equivalent to the event that any part of $p(i,j)$ is not in $R_{G_B}$.
Apply Eq. (\ref{pr-1}) to $C=p(i,j)$ and obtain Eq. (\ref{sl}).
\end{proof}

Surprisingly, $\Pr (i\leftrightarrow j)$ can be given by a simple explicit function of $p(i,j)$, and is determined only by $\lengthx{\overline{p(i,j)}}$. That is, for fixed $\lengthx{\overline{p(i,j)}}$, $\Pr (i\leftrightarrow j)$ is not dependent on the length of $p(i,j)$ or the size of $\overline{p(i,j)}$.

The result mentioned above tells us the following.
First, if there is an intermediate node $n_k$ between $n_i$ and $n_j$ (Fig. \ref{single-route}), $\Pr(i\leftrightarrow k)-\Pr(i\leftrightarrow j)$ is proportional to $\lengthx{\overline{p(i,j)}}-\lengthx{\overline{p(i,k)}}$.
That is, the decrease in the probability of maintaining connectivity due to the increase in the number of hops is proportional to the increase in the perimeter length of the convex hull of the route.
Therefore, if the increase in the perimeter length of the convex hull of the route is not large, the decrease in the probability of maintaining connectivity due to the increase in the number of hops is not large.
For example, $\Pr(i\leftrightarrow k)-\Pr(i\leftrightarrow j)$ is larger in Fig. \ref{single-route}-(a) than in Fig. \ref{single-route}-(b) because the increase in the perimeter length of the convex hull of the route is larger in the former than in the latter.
Second, $\Pr(i\leftrightarrow j)$ is smaller as $\lengthx{\overline{p(i,j)}}$ increases.
Hence, it is likely that $\Pr(i\leftrightarrow j)$ is smaller if $n_i$ and $n_j$ are far away or the route between them is more non-roundabout (Fig. \ref{single-route}-(a)) than roundabout (Fig. \ref{single-route}-(b)) for fixed route length.
Third, among all $i,j$ pairs, $\Pr (i\leftrightarrow j)$ becomes worse for the $i,j$ pair of which $\lengthx{\overline{p(i,j)}}$ is the largest.

\begin{figure}[htb] 
\begin{center} 
\includegraphics[width=8cm,clip]{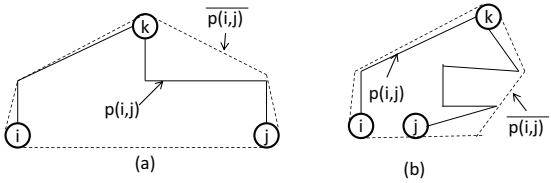} 
\caption{$\Pr(i\leftrightarrow j)$ for single route} 
\label{single-route} 
\end{center} 
\end{figure}

Equation (\ref{sl}) can provide $\Pr(i\leftrightarrow k\leftrightarrow j)$ by replacing $p(i,j)$ with $p(i,k)\cup p(k,j)$ if $p(i,k)$ and $p(k,j)$ are single route paths.
Particularly when $n_k\in p(i,j)$, $p(i,j)=p(i,k)\cup p(k,j)$.

\begin{result}
Between $n_i$ and $n_j$, there is a single route for $j=j_1, j_2, \cdots$.
Let $P_a$ be the probability that there is connectivity between $n_i$ and all these nodes $n_{j_1},n_{j_2},\cdots$.
Then, 
\begin{equation}
P_a=\frac{\lengthx{A_0}-\lengthx{\overline{\cup_{j=j_1, j_2, \cdots}  p(i,j)}}}{2\lengthx{A_0}+2\pi w}.\label{P_l}
\end{equation}
\end{result}

\begin{proof}
The event $\cap_{j=j_1, j_2, \cdots}\{i\leftrightarrow j\}$ is equivalent to $(\cup_{j=j_1, j_2, \cdots}p(i,j))\cap R_{G_B}=\emptyset$.
According to Eq. (\ref{pr-1}), we obtain Eq. (\ref{P_l}).
\end{proof}
It is likely that $\lengthx{\overline{\cup_{j=j_1, j_2, \cdots}  p(i,j)}}$ is large when a network covers a physically wide area.
Thus, it is difficult to maintain the connectivity of the entire network during a disaster if the network covers a physically wide area.
The difference in $P_a$ for a small network and that for a large network is proportional to the difference in $\lengthx{\overline{\cup_{j=j_1, j_2, \cdots}  p(i,j)}}$ for these two networks.

\subsection{Ring-type network}
\begin{result}\label{ring-type-path}
Assume that $p(i,j)$ is a ring-type network.
Nodes $n_i$ and $n_j$ are connected if at least one clockwise route or counterclockwise route of this network is maintained.
Then, $\{n_i\leftrightarrow n_j\}$ if and only if $\{\overline{(n_i,n_j)}\cap R_{G_B}=\emptyset\}$, and its probability is given as follows. 
\begin{equation}
\Pr(i \leftrightarrow j)=\frac{\lengthx{A_0}-\lengthx{\overline{(n_i,n_j)}}}{2\lengthx{A_0}+2\pi w}.\label{xy}
\end{equation}
This probability does not increase even if additional routes are provided between nodes on this ring-type network (Fig. \ref{high-connect}).
\end{result}

\begin{figure}[htb] 
\begin{center} 
\includegraphics[width=8cm,clip]{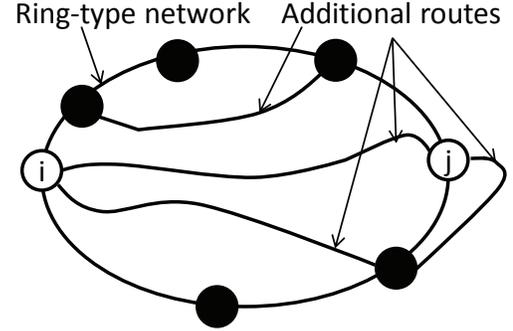} 
\caption{Higher connectivity by additional routes} 
\label{high-connect} 
\end{center} 
\end{figure}

\begin{proof}
First, assume that $p(i,j)$ is a ring-type network. That is, no additional routes are assumed to be provided. If $\{i\not\leftrightarrow j\}$, (i) $n_i$ and $n_j$ are in $R_{G_B}$, (ii) either $n_i$ or $n_j$ is in $R_{G_B}$, or (iii) neither $n_i$ nor $n_j$ is in $R_{G_B}$, but there exist a link  in $R_{G_B}$ on the clockwise route and another link in $R_{G_B}$ on the counter-clockwise route.
Due to the assumption of the convexity of $p(i,j)$ and the fact that $R_{G_B}$ is a half plane, case (iii) does not occur.
It is also clear that $\{\overline{(n_i,n_j)}\cap R_{G_B}\neq\emptyset\}$ for cases (i) and (ii).
Thus, $\{i\not\leftrightarrow j\}\Rightarrow \{\overline{(n_i,n_j)}\cap R_{G_B}\neq\emptyset\}$.
Equivalently, $\{\overline{(n_i,n_j)}\cap R_{G_B}=\emptyset\}\Rightarrow\{i\leftrightarrow j\}$.

Conversely, if $\{\overline{(n_i,n_j)}\cap R_{G_B}\neq\emptyset\}$, $\{\overline{(n_i,n_j)}\subset R_{G_B}\}$ or $\{\overline{(n_i,n_j)}\cap G_B\neq \emptyset\}$.
Note $\{\overline{(n_i,n_j)}\subset R_{G_B}\}$ means that $\{i\not\leftrightarrow j\}$.
In addition, $\{\overline{(n_i,n_j)}\cap G_B\neq \emptyset\}$ means that either $n_i$ or $n_j$ is in $R_{G_B}$ and $\{i\not\leftrightarrow j\}$.
Therefore, $\{\overline{(n_i,n_j)}\cap R_{G_B}\neq\emptyset\}\Rightarrow\{i\not\leftrightarrow j\}$.
Consequently, $\{i\leftrightarrow j\}\Rightarrow \{\overline{(n_i,n_j)}\cap R_{G_B}=\emptyset\}$.

Hence, $\{i\leftrightarrow j\}$ if and only if $\{\overline{(n_i,n_j)}\cap R_{G_B}=\emptyset\}$.

According to Eq. (\ref{pr-1}) with $C=\overline{(n_i,n_j)}$, $\Pr(i\leftrightarrow j)=\frac{\lengthx{A_0}-\lengthx{\overline{(n_i,n_j)}}}{2\lengthx{A_0}+2\pi w}$.

Second, assume that additional routes are provided.
It is almost clear that additional routes do not increase $\Pr(n_i\leftrightarrow n_j)$.
If $i\not\leftrightarrow j$ without additional routes, at least one of $n_i,n_j$ is in $R_{G_B}$.
This is because case (iii) does not occur, as mentioned above.
Therefore, no additional routes increase $\Pr(n_i\leftrightarrow n_j)$, because at least one of $n_i,n_j$ is in $R_{G_B}$. 
\end{proof}

Equation (\ref{xy}) tells us that we cannot change $\Pr(i\leftrightarrow j)$ by changing the physical route of the ring-type network for fixed locations of $n_i$ and $n_j$.
This is because the right-hand side of Eq. (\ref{xy}) does not depend on $p(i,j)$.

Comparing Eq. (\ref{xy}) with Eq. (\ref{sl}), we find that the effect of the clockwise and counterclockwise routes on $\Pr (i\leftrightarrow j)$ for a ring-type network is equal to the replacement of the ring-type $p(i,j)$ with the single straight-line-segment physical route between $n_i$ and $n_j$.
This is actually intuitive.
Due to the convexity of the ring-type network $p(i,j)$ and the fact that the boundary of the disaster area model $R_{G_B}$ is a line, the event that the ring-type network $p(i,j)$ intersects $R_{G_B}$ is identical to the event that the straight-line segment between $n_i$ and $n_j$ intersects $R_{G_B}$ (Fig. \ref{convex-ring}).

\begin{figure}[htb] 
\begin{center} 
\includegraphics[width=8cm,clip]{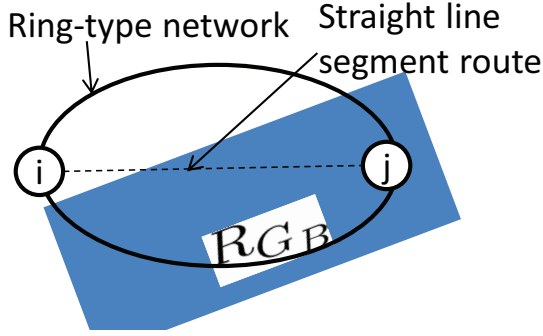} 
\caption{Ring-type network intersecting $D=R_{G_B}$} 
\label{convex-ring} 
\end{center} 
\end{figure}

\subsection{Combination of single route path and ring-type network}
An actual network is not as simple as a tree or a ring, and uses mechanisms to improve the availability of the network.
In this subsection, the results on more practical situations are discussed.

The latter part of the following result is useful to evaluate cases in which $n_k$ is equipped with a higher-layer switch or a server.

\begin{result}\label{comb1}
Assume that path configuration $p(i,j)$ between $n_i$ and $n_j$ consists of two parts: the first part $p(i,k)$ between $n_i$ and $n_k$ is a single route path, and the second part $p(k,j)$ between $n_k$ and $n_j$ is a ring-type network where $p(k,j)$ is convex (Fig. \ref{combination}-(a)).
The probability $\Pr(i\leftrightarrow j)$ that there is connectivity between $n_i$ and $n_j$ is given by  
\begin{equation}
\Pr (i\leftrightarrow j)=\frac{\lengthx{A_0}-\lengthx{\overline{p(i,k)\cup n_j}}}{2\lengthx{A_0}+2\pi w}.\label{sl2}
\end{equation}

Assume that path configuration $p(i,j)$ between $n_i$ and $n_j$ consists of four parts: the first part $p(i,k_1)$ between $n_i$ and $n_{k_1}$ and the fourth part $p(k_2,j)$ between $n_{k_2}$ and $n_j$ are single route paths, and the second and third parts $p(k_1,k)$ and $p(k,k_2)$ are convex ring-type networks on the same ring network (Fig. \ref{combination}-(b)).
The probability $\Pr(i\leftrightarrow j)$ is given by  
\begin{eqnarray}
\Pr(i\leftrightarrow j)
&=&\frac{\lengthx{A_0}-\lengthx{\overline{p(i,k_1)\cup n_k\cup p(k_2,j)}}}{2\lengthx{A_0}+2\pi w }\label{pij2}
\end{eqnarray}

\end{result}
\begin{proof}
The event $\{i\leftrightarrow j\}$ is equivalent to the joint event of $\{i\leftrightarrow k\}$ and $\{k\leftrightarrow j\}$ for case (a).
The event $\{i\leftrightarrow k\}$ is equivalent to $\{p(i,k)\cap R_{G_B}=\emptyset\}$.
Due to the definition of the convex hull, this is equivalent to $\{\overline{p(i,k)}\cap R_{G_B}=\emptyset\}$.
The event $\{k\leftrightarrow j\}$ is equivalent to $\{\overline{(n_k,n_j)}\cap R_{G_B}=\emptyset\}$ because of Result \ref{ring-type-path}.
Note that $\{\overline{p(i,k)}\cap R_{G_B}=\emptyset\}\cap\{\overline{(n_k,n_j)}\cap R_{G_B}=\emptyset\}$ is equivalent to $\{\overline{p(i,k)\cup n_j}\cap R_{G_B}=\emptyset\}$.
By applying Eq. (\ref{pr-1}), we obtain Eq. (\ref{sl2}).

The event $i\leftrightarrow k_1\leftrightarrow k \leftrightarrow k_2\leftrightarrow j$ is equivalent to the event $\{p(i,k_1)\cap R_{G_B}=\emptyset\}\cap\{k_1\leftrightarrow k\}\cap\{k\leftrightarrow k_2\}\cap\{p(k_2,j)\cap R_{G_B}=\emptyset\}$.
According to Result \ref{ring-type-path}, $\{k_1\leftrightarrow k\}$ is equivalent to $\{\overline{(n_{k_1},n_k)}\cap R_{G_B}=\emptyset\}$.
That is, $i\leftrightarrow k_1\leftrightarrow k \leftrightarrow k_2\leftrightarrow j$ is equivalent to $\{p(i,k_1)\cup p(j,k_2)\cup \overline{(n_{k_1},n_k)}\cup \overline{(n_k,n_{k_2})}\}\cap R_{G_B}=\emptyset$.
According to Eq. (\ref{pr-1}), we obtain Eq. (\ref{pij2}).
\end{proof}

\begin{figure}[htb] 
\begin{center} 
\includegraphics[width=8cm,clip]{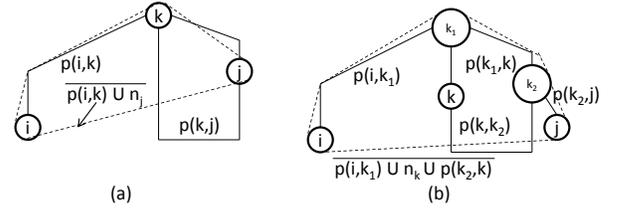} 
\caption{$\Pr(i\leftrightarrow j)$ for combination of single route path and ring-type network} 
\label{combination} 
\end{center} 
\end{figure}

\begin{result}
Assume that path configuration $p(i,j)$ between $n_i$ and $n_j$ consists of two parts: the first part $p(i, k(j))$ is a ring-type network between $n_i$ and $n_{k(j)}$, and the second part $p(k(j),j)$ is a single route for $j=j_1, j_2, \cdots$.
Let $P_a$ be the probability that there is connectivity between $n_i$ and all of these nodes $n_{j_1},n_{j_2},\cdots$.
Then, 
\begin{equation}
P_a=\frac{\lengthx{A_0}-\lengthx{\overline{\cup_{j=j_1, j_2, \cdots}p(k(j),j)\cup n_i}}}{2\lengthx{A_0}+2\pi w},\label{P_r}
\end{equation}

\end{result}
\begin{proof}
The event $\cap_{j=j_1, j_2, \cdots}\{i\leftrightarrow j\}$ is equivalent to $(\cup_{j=j_1, j_2, \cdots}p(k(j),j)\cup\overline{(k(j),i)})\cap R_{G_B}=\emptyset$.
By using Eq. (\ref{pr-1})  and $\cup_{j=j_1, j_2, \cdots}p(k(j),j)\cup\overline{(k(j),i)}=\overline{\cup_{j=j_1, j_2, \cdots}p(k(j),j)\cup n_i}$, we obtain Eq. (\ref{P_r}).
\end{proof}

How using a backup node $n_{j_2}$ improves the probability of maintaining connectivity is now discussed.

\begin{result}\label{backup1}
Assume that there are two nodes $n_{j_1}$ and $n_{j_2}$ on a ring network.
The path configurations $p(i,j_1)$ and $p(i,j_2)$ consist of two parts: the first part ($p(i,k)$) is a single route path, and the second part ($p(k,j_1)$ or $p(k,j_2)$) is a ring-type network.
The path configurations $p(k,j_1)$ and $p(k,j_2)$ form the same convex ring-type network.
The probability $\Pr (i\leftrightarrow j_1\ {\rm or}\ j_2)$ that there is connectivity between $n_i$ and at least one of $n_{j_1}$ and $n_{j_2}$ is given by  
\begin{eqnarray}
&&\Pr (i\leftrightarrow j_1\ {\rm or}\ j_2)\cr
&=&\frac{1}{2\lengthx{A_0}+2\pi w}\{\lengthx{A_0}-\lengthx{\overline{p(i,k)\cup n_{j_1}}}-\lengthx{\overline{p(i,k)\cup n_{j_2}}}\cr
&&\quad\quad\quad\quad\quad\quad+\lengthx{\overline{p(i,k)\cup n_{j_1} \cup n_{j_2}}}\}.\label{two-lc}
\end{eqnarray}

Assume that there are nodes $n_{j_1},n_{j_2},n_{j_3},\cdots, n_k$ on a ring network and that $n_{j_1}$ and $n_{j_2}$ are the nearest destination nodes on the clockwise route and counterclockwise route on the ring network from $n_i$, respectively (Fig. \ref{backup1}). Then,
\begin{equation}
\Pr (i\leftrightarrow j_1, j_2,\cdots, {\rm or}\ j_k)=\Pr (i\leftrightarrow j_1\ {\rm or}\ j_2).\label{multiple-lc}
\end{equation}

\end{result}

\begin{figure}[htb] 
\begin{center} 
\includegraphics[width=8cm,clip]{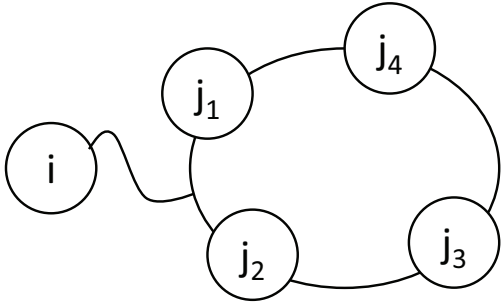} 
\caption{$\Pr (i\leftrightarrow j_1, j_2,\cdots, {\rm or}\ j_k)$}
\label{backup1} 
\end{center} 
\end{figure}

\begin{proof}
The event that there is connectivity between $n_i$ and one of $n_{j_1}$ and $n_{j_2}$ is equivalent to the event $ \{ (p(i,k)\cup \overline{(n_k,n_{j_1})}) \cap R_{G_B}=\emptyset\}\ {\rm or}\ \{ (p(i,k)\cup \overline{(n_k,n_{j_2})})\cap R_{G_B}=\emptyset\}$.
To evaluate the measure of this event, we need to divide it into two subevents and remove the overlap: the first subevent is $\{(p(i,k)\cup \overline{(n_k,n_{j_1})})\cap R_{G_B}=\emptyset\}$ and the second subevent is $\{(p(i,k)\cup \overline{(n_k,n_{j_2})})\cap R_{G_B}=\emptyset\}$.
Because their overlap is $\{(p(i,k)\cup\overline{(n_k,n_{j_1})}\cup \overline{(n_k,n_{j_2})})  \cap R_{G_B}=\emptyset\}$, $m(B; i\leftrightarrow j_1\ {\rm or}\ j_2)=m(B; (p(i,k)\cup \overline{(n_k,n_{j_1})})\cap R_{G_B}=\emptyset)+m(B; (p(i,k)\cup \overline{(n_k,n_{j_2})})\cap R_{G_B}=\emptyset)-m(B; (p(i,k)\cup\overline{(n_k,n_{j_1})}\cup \overline{(n_k,n_{j_2})}) \} \cap R_{G_B}=\emptyset$.
Similar to Eq. (\ref{pr-1}), we obtain Eq. (\ref{two-lc}).

If there is connectivity between $n_i$ and one of nodes $n_{j_3},n_{j_4},\cdots$, there is connectivity between $n_i$ and at least $n_{j_1}$ and $n_{j_2}$.
Therefore, Eq. (\ref{multiple-lc}) is valid.
\end{proof}

Equation (\ref{two-lc}) tells us that the effect of a backup node on $\Pr (i\leftrightarrow j_1\ {\rm or}\ j_2)$ is $\{-\lengthx{\overline{p(i,k)\cup n_{j_2}}}+\lengthx{\overline{p(i,k)\cup n_{j_1} \cup n_{j_2}}}\}/(2\lengthx{A_0}+2\pi w)$.
This effect is numerically evaluated with numerical examples discussed in a later section.
Equation (\ref{multiple-lc}) shows that more than two backup nodes are meaningless for a given $i$.

The following result is an extension of the latter part of Result \ref{comb1}.
That is, it covers cases in which a higher layer switch or a server at $n_{k_1}$ has a backup $n_{k_2}$.

\begin{result}
Assume that there are two nodes $n_{k_1}$ and $n_{k_2}$ on a ring network.
To connect between $n_i$ and $n_j$, one of them must be connected.
The path configuration $p(i,j)$ consists of three parts (Fig. \ref{backup}): the first and third parts $p(i,i_0)$ $p(j_0,j)$ are single route paths, and the second part $p(i_0,j_0)$ is a little bit complicated.
When $n_{k_1}$ ($n_{k_2}$) is used, the second part $p(i_0,j_0)$ consists of two ring-type networks $p(i_0,k_1)$ and $p(k_1,j_0)$ ($p(i_0,k_2)$ and $p(k_2,j_0)$).
The probability that there is connectivity between $n_i$ and $n_j$ through either $n_{k_1}$ or $n_{k_2}$ is given as follows.
\begin{eqnarray}
&&\Pr(i\leftrightarrow i_0\leftrightarrow (k_1\ {\rm or}\ k_2) \leftrightarrow j_0\leftrightarrow j)\cr
&=&\frac{1}{2\lengthx{A_0}+2\pi w}\{\lengthx{A_0}-\lengthx{\overline{p(i,i_0)\cup p(j,j_0)\cup n_{k_1}}}\cr
&&\qquad\qquad\qquad-\lengthx{\overline{p(i,i_0)\cup p(j,j_0)\cup n_{k_2}}}\cr
&&\qquad\qquad+\lengthx{\overline{p(i,i_0)\cup p(j,j_0)\cup n_{k_1}\cup n_{k_2}}}\}\label{pij22}
\end{eqnarray} 
\end{result}

\begin{figure}[htb] 
\begin{center} 
\includegraphics[width=8cm,clip]{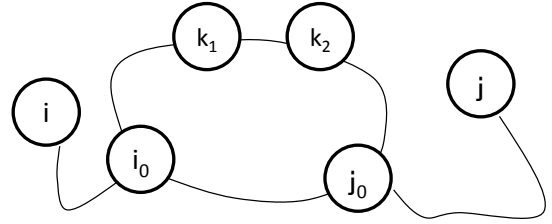} 
\caption{$\Pr (i\leftrightarrow j)$ through either $n_{k_1}$ or $n_{k_2}$} 
\label{backup} 
\end{center} 
\end{figure}

\begin{proof}
The event in which there is connectivity between $n_i$ and $n_j$ through either $n_{k_1}$ or $n_{k_2}$ is equivalent to $\{ i\leftrightarrow i_0\leftrightarrow (k_1\ {\rm or}\ k_2) \leftrightarrow j_0\leftrightarrow j\}$.
To evaluate the measure of this event, divide this event into two subevents and remove the overlap: the first subevent is $\{ i\leftrightarrow i_0\leftrightarrow k_1 \leftrightarrow j_0\leftrightarrow j\}$ and the second subevent is $\{ i\leftrightarrow i_0 \leftrightarrow k_2 \leftrightarrow j_0\leftrightarrow j\}$.
The measures of these two subevents are given by the numerator of Eq. (\ref{pij2}).
The overlap is $\{ i\leftrightarrow i_0\leftrightarrow k_1 \leftrightarrow j_0\leftrightarrow j\}\cap\{ i\leftrightarrow i_0 \leftrightarrow k_2 \leftrightarrow j_0\leftrightarrow j\}$.
This is identical to $\{(p(i,i_0)\cup p(j,j_0)\cup  n_{k_1}\cup n_{k_2})\cap R_{G_B}=\emptyset$.
Therefore, the measure of this overlap is $(\lengthx{A_0}-\lengthx{\overline{p(i,i_0)\cup p(j,j_0)\cup n_{k_1}\cup n_{k_2}}})$.
As a result, we obtain Eq. (\ref{pij22}).
\end{proof}
We learn that the effect of a backup node on the probability that there is connectivity between $n_i$ and $n_j$ through at least one $n_{k_1}$ or $n_{k_2}$ is quite similar to that on $\Pr (i\leftrightarrow j_1\ {\rm or}\ j_2)$.
Replacing a single route path $p(i,k)$ in Eq. (\ref{two-lc}) with the union of two single route paths $p(i,i_0)\cup p(j,j_0)$ gives us Eq. (\ref{pij22}).

\subsection{Number of nodes losing connectivity}
Let us discuss how we minimize the expected number of subscribers or nodes losing connectivity.

\begin{result}
Let $S$ be a set of nodes and $\sharp(i \not\leftrightarrow  S)$ be the number of nodes in $S$ disconnected from $n_i$.
If $p(i,j)$ is a single route path for $n_j\in S$,
\begin{eqnarray}
E[\sharp(i \not\leftrightarrow  S)]
&=&\frac{\sum_{n_j\in S} (\lengthx{A_0}+\lengthx{\overline{p(i,j)}}+2\pi w)}{2\lengthx{A_0}+2\pi w}.\label{number}\cr
&&
\end{eqnarray}
If $p(i,j)$ is a ring-type network for $n_j\in S$,
\begin{eqnarray}
E[\sharp(i \not\leftrightarrow  S)]
&=&\frac{\sum_{n_j\in S} (\lengthx{A_0}+\lengthx{\overline{(n_i,n_j)}}+2\pi w)}{2\lengthx{A_0}+2\pi w}.\label{number2}\cr
&&
\end{eqnarray}
\end{result}
\begin{proof}
Note that $E[\sharp(i \not\leftrightarrow  S)]=\sum_{n_j\in S} E[\bfone (i\not\leftrightarrow j)]=\sum_{n_j\in S} \Pr(i\not\leftrightarrow j)$ where $\bfone(x)=\cases{1, &if $x$ is true, \cr 0, &otherwise.}$
Then, we obtain Eq. (\ref{number}) through Eq. (\ref{sl}).
Similarly, we can obtain Eq. (\ref{number2}) by using Eq. (\ref{xy}).
\end{proof}
Therefore, to reduce $E[\sharp(i \not\leftrightarrow  S)]$, reduction of $\sum_j\lengthx{\overline{p(i,j)}}$ is necessary for a single route.
Thus, its physical route that minimizes $E[\sharp(i \not\leftrightarrow  S)]$ is identical to every $p(i,j)$ being on a straight line when the locations of the nodes are fixed.
In practice, because there are many constraints, such as alignment with the road and minimizing cost (including minimizing the total length of the cable and maximizing the number of fibers in use in a cable) in determining the physical route, the straight-line route is difficult.
By choosing the physical route that minimizes $\lengthx{\overline{p(i,j)}}$ among the feasible routes, however, we can reduce $E[\sharp(i \not\leftrightarrow  S)]$.
In particular, even when the length of the route is the same, we can reduce $\lengthx{\overline{p(i,j)}}$ and, as a result, $E[\sharp(i \not\leftrightarrow  S)]$.
Figure \ref{example} shows such an example. The lengths of Routes 1 and 2 are the same, but the perimeter length of the convex hull of Route 1 is shorter than that of Route 2.

\begin{figure}[htb] 
\begin{center} 
\includegraphics[width=8cm,clip]{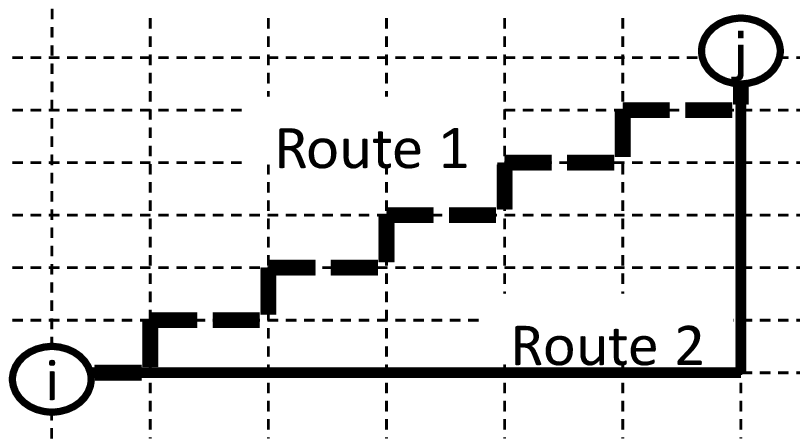} 
\caption{Example of reducing $E[\sharp(i \not\leftrightarrow  S)]$ with same route length} 
\label{example} 
\end{center} 
\end{figure}

A very good point of Eq. (\ref{number}) is that we can reduce $E[\sharp(i \not\leftrightarrow  S)]$ by reducing $\lengthx{\overline{p(i,j)}}$ independently of the path configuration of other source-destination pairs.
Hence, reducing $\lengthx{\overline{p(i,j)}}$ (as a result, reducing $E[\sharp(i \not\leftrightarrow  S)]$) can be implemented easily without burden of computation.

On the other hand, as described just after Result \ref{ring-type-path}, $E[\sharp(i \not\leftrightarrow  S)]$ does not depend on the physical route if the locations of $n_i$ and $n_j$ are fixed and $p(i,j)$ is a ring-type network.
Thus, we cannot reduce $E[\sharp(i \not\leftrightarrow  S)]$ by changing the physical route.

Although the above result does not refer to other types of path configuration other than a single route path or a ring-type network, it is trivial that we can obtain $E[\sharp(i \not\leftrightarrow  S)]$ if we know $\Pr(i \not\leftrightarrow  j)$ because $E[\sharp(i \not\leftrightarrow  S)]=\sum_{n_j\in S} \Pr(i\not\leftrightarrow j)$.

\section{Numerical examples}\label{large-n}
In the following numerical examples, $w=0$ is assumed for any case if not explicitly indicated otherwise.
This is because we can easily obtain the probability connecting two nodes for $w>0$ by $\frac{\lengthx{A_0}}{\lengthx{A_0}+2\pi w}\times$ (the probability connecting two nodes for $w=0$).

\subsection{Subscriber network}
By using Eq. (\ref{sl}), $\Pr(i\leftrightarrow l_c\leftrightarrow j)$ is evaluated for the network model shown in Fig. \ref{subs_models}-(a).
Here, $l_c$ is a local network center providing a local switch to connect two subscribers in this subscriber network, and is located at the bottom-left corner in this figure.
In addition, $\Pr(i\leftrightarrow l_c\leftrightarrow j)$ can be evaluated for the network model shown in Fig. \ref{subs_models}-(b) by using Eq. (\ref{pij2}).
It can also evaluated for the real subscriber network shown in Fig. \ref{real-subscriberNW}.

\begin{figure}[htb] 
\begin{center} 
\includegraphics[width=8cm,clip]{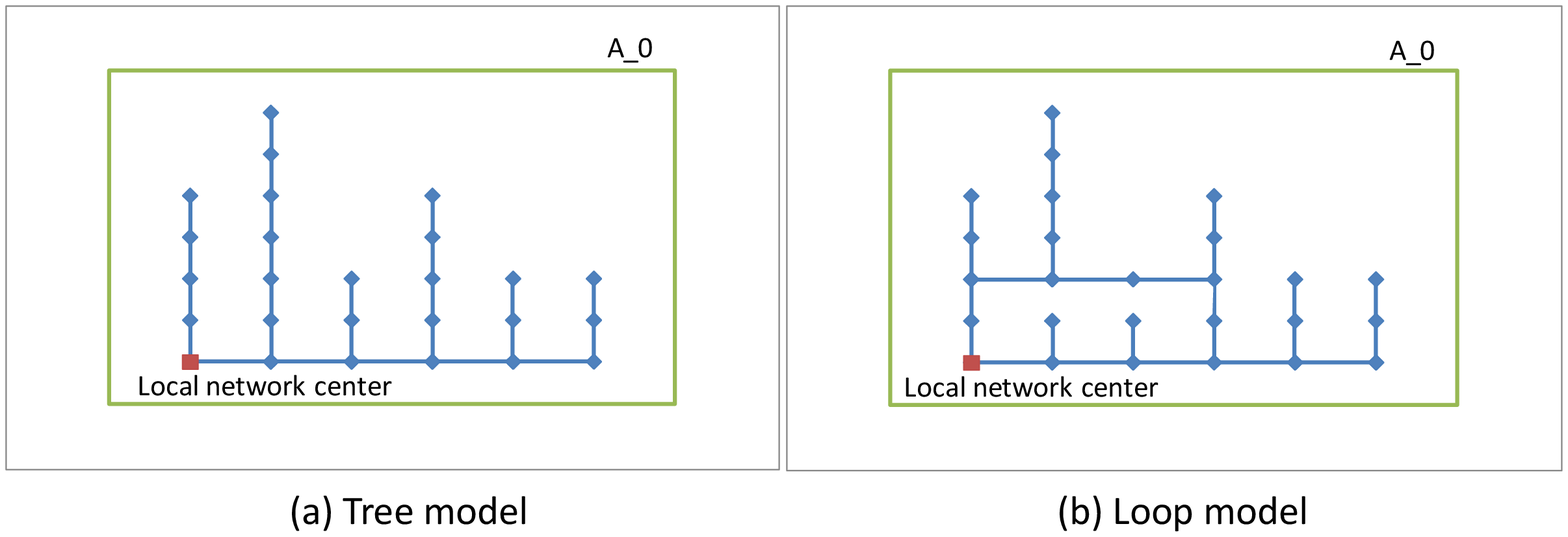} 
\caption{Subscriber network models} 
\label{subs_models} 
\end{center} 
\end{figure}

\begin{figure}[htb] 
\begin{center} 
\includegraphics[width=8cm,clip]{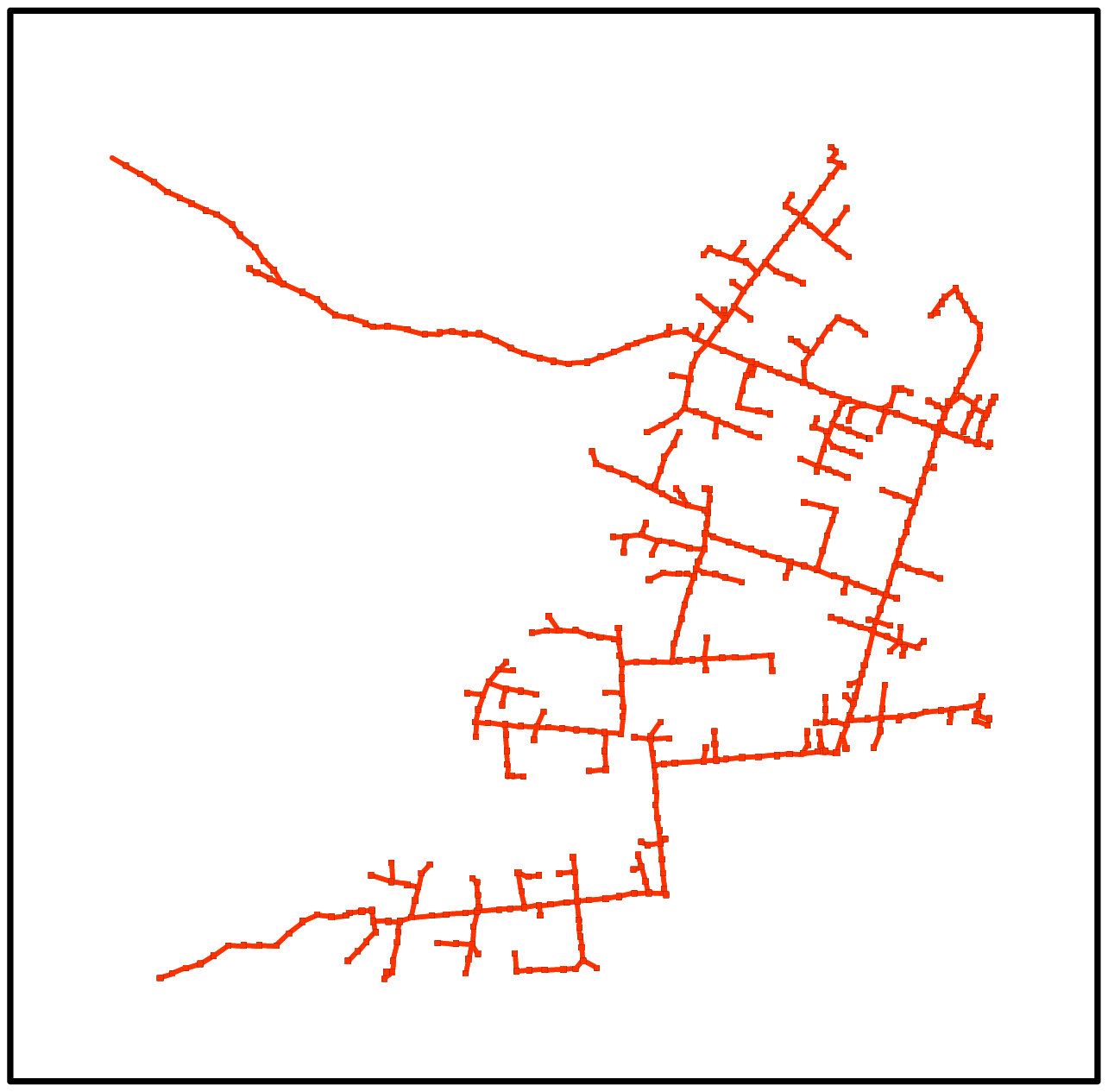} 
\caption{Example of real subscriber network} 
\label{real-subscriberNW} 
\end{center} 
\end{figure}

Figures \ref{subs_model_result} and \ref{real-sub-result} plot $\Pr(i\leftrightarrow l_c\leftrightarrow j)$ after sorting them in descending order of their exact values.
In these figures, ^^ ^^ Exact" is given by Eq. (\ref{sl}) for (a) and by Eq. (\ref{pij2}) for (b), ^^ ^^ Independent approx." is given by $\Pr(i\leftrightarrow l_c)\Pr(j\leftrightarrow l_c)$, each of which is derived by Eq. (\ref{sl}) for (a) and by Eq. (\ref{sl2}) for (b), ^^ ^^ Line approx." is given by Eq. (\ref{sl}) under the assumption that the physical route between $n_i$ and $n_j$ is given by a straight line (a flybird connection \cite{itc}, \cite{access}, so to speak), and ^^ ^^ Independent line approx." is given by $\Pr(i\leftrightarrow l_c)\Pr(j\leftrightarrow l_c)$ under the assumption that the physical route between $n_i$ ($n_j$) and $l_c$ is given by a straight line.
A simulation was conducted assuming that $D=R_{G_B}$ to confirm the verification of ^^ ^^ Exact."
Because the simulation results were in very good agreement with ^^ ^^ Exact", we cannot see ^^ ^^ Exact" overlapping with the simulation results.
(The average sizes of the 95\% confidence intervals of the simulation are 0.0182, 0.0183, and 0.0190 for the network shown in Figs. \ref{subs_models}-(a), (b), and \ref{real-subscriberNW}.)

 As expected, the independent approximation underestimates and the line approximation overestimates $\Pr(i\leftrightarrow l_c\leftrightarrow j)$.
This is because the independent approximation ignores the correlations that the two routes (between $n_i$ and $l_c$ and between $n_j$ and $l_c$) are damaged simultaneously and because the line approximation ignores the fact that the actual physical route is longer than the straight-line route.
Independent line approximation is almost similar to independent approximation.
Line approximation shows us the minimum $\Pr(i\leftrightarrow l_c\leftrightarrow j)$ by changing the physical route.
Its effect is not so large in this figure.
In addition, the effect of the introduction of the local loop can be evaluated by the difference between $\Pr(i\leftrightarrow l_c\leftrightarrow j)$ under (a) the tree model and that under (b) the loop model for each pair of $i$ and $j$.
Of course, the difference is always non-negative, but it is very small in this example.
It is less than three percent (its relative value is less than ten percent) for any pair of $i$ and $j$, and it is less than one percent on average (its relative value is less than two percent on average).

\begin{figure}[htb] 
\begin{center} 
\includegraphics[width=8cm,clip]{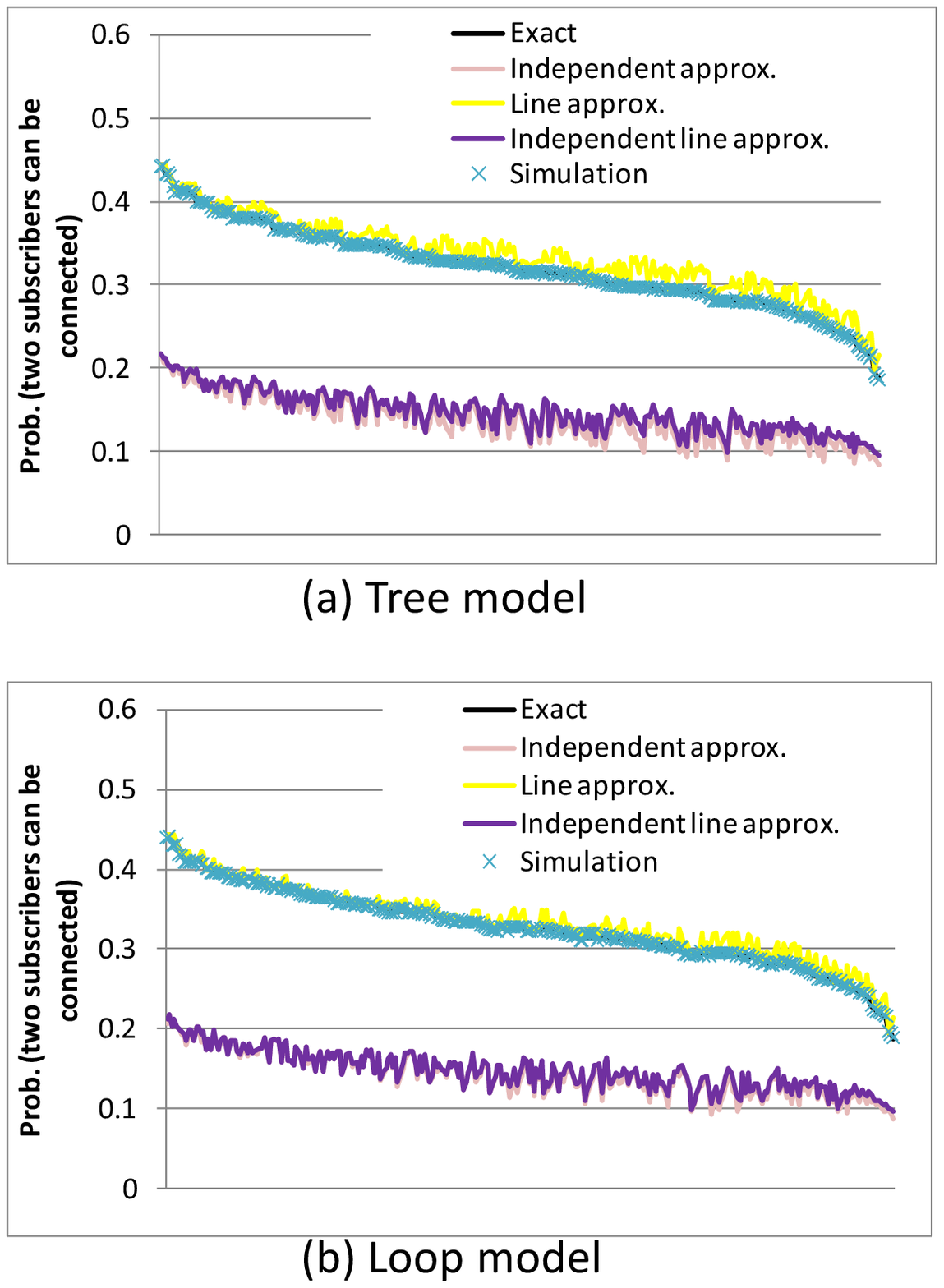} 
\caption{$\Pr(i\leftrightarrow l_c\leftrightarrow j)$ for subscriber network models} 
\label{subs_model_result} 
\end{center} 
\end{figure}

\begin{figure}[htb] 
\begin{center} 
\includegraphics[width=8cm,clip]{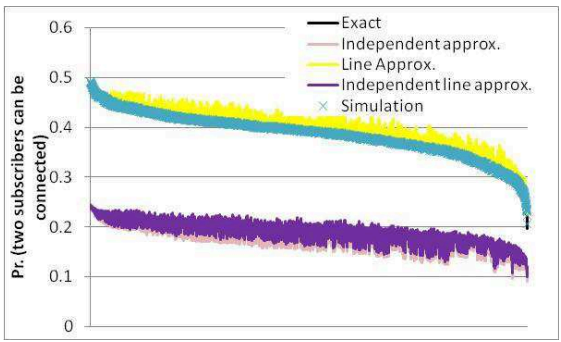} 
\caption{$\Pr(i\leftrightarrow l_c\leftrightarrow j)$ for real subscriber network} 
\label{real-sub-result} 
\end{center} 
\end{figure}

\subsection{Effect of backup node}
Based on Eqs. (\ref{sl2}) and (\ref{two-lc}), we evaluated the effect of a backup node for the model shown in Fig. \ref{two-center}.
A regional network center $r_c$ is located on a ring network, which is a circle with radius $r_l$. The angle formed by the line segment between $r_c$ and the center $O$ of the ring network and the reference line is $\alpha_1$.
When the other regional network center $r'_c$ as a backup is provided, the angle of the line passing $O$ and $r'_c$ and the reference line is $\alpha_2$.
There are multiple nodes $n_1,n_2,\cdots$ on the same ring network, where the angle between two consecutive nodes is $\gamma$.
Let $S$ be the set of these nodes.

Similar to derivation of Eq. (\ref{number}), $E1\equiv E[\sharp(r_c \not\leftrightarrow  S)]=\sum_{n_j\in S} (1-\Pr(r_c\leftrightarrow j))$ and $E2\equiv E[\sharp(r_c\ {\rm or}\ r_c' \not\leftrightarrow  S)]=\sum_{n_j\in S} (1-\Pr(r_c\ {\rm or}\ r_c' \leftrightarrow j))$.
Because $\Pr(r_c \leftrightarrow j)$ is given by Eq. (\ref{sl2}) with $r_c=i=k$ and $\Pr(r_c\ {\rm or}\ r_c' \leftrightarrow j))$ is given by Eq. (\ref{two-lc}) with $j=i=k$ and $r_c=j_1, r_c'=j_2$, we can numerically minimize $E1$ by changing the location of $r_c$ and $E2$ by changing the locations of $r_c$ and $r_c'$.

\begin{figure}[htb] 
\begin{center} 
\includegraphics[width=8cm,clip]{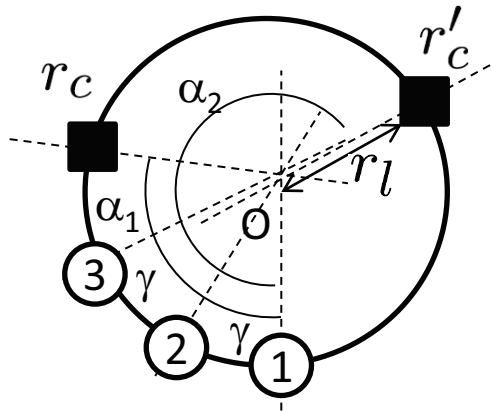} 
\caption{Model of two regional network centers} 
\label{two-center} 
\end{center} 
\end{figure}

$E1^*$ and $E2^*$, which are minimized $E1$ and $E2$, are plotted in Figure \ref{result-two-center}.
In this figure, (i) as the number of nodes increases, $E1^*-E2^*$ almost linearly increases, and (ii) as $r_l$ increases, $E1^*-E2^*$ slightly increases.
Therefore, for a network with many nodes or for a physically large network, we need to provide a backup node.
Although it can be determined whether the effect of a backup node is large because $E1^*-E2^*$ and $E1^*/E2^*$ depends on $\lengthx{A_0}$, the backup node does not seem able to drastically reduce the number of nodes disconnected from a regional network center.

\begin{figure}[htb] 
\begin{center} 
\includegraphics[width=8cm,clip]{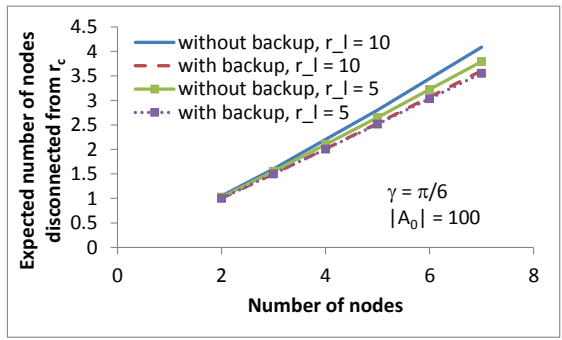} 
\caption{Effect of two regional network centers} 
\label{result-two-center} 
\end{center} 
\end{figure}

\subsection{Validity of strip model}\label{strip-section}
The boundary area of a disaster area $D$ is assumed to be a strip $B$ in the numerical examples in this section when $D$ is much larger than the distance between two nodes of interest.
This assumption may be valid in a macroscopic view, but it is not likely in a microscopic view.

To evaluate the impact of the micro shape of the boundary, a simulation was conducted that uses a sine wave as the boundary of $D$.
Figure \ref{error_wave} plots the mean relative absolute error $\sum_{i\neq j}e(i,j)/\sum_{i\neq j}1$ for two cases (tree model in Fig. (\ref{subs_models}) and real subscriber network shown in Fig. \ref{real-subscriberNW}), where $e(i,j)$ is the relative absolute error defined by $e(i,j)\equiv |p_{simulation}(i,j)-p_{theory}(i,j)|/p_{theory}(i,j)$. 
Here, $p_{theory}(i,j)$ is $\Pr(i\leftrightarrow l_c \leftrightarrow j)$ obtained by Eq. (\ref{sl}), and $p_{simulation}(i,j)$ is that obtained by simulation when the sine wave with amplitude $w/2$ and wavelength $\lambda$ is used as the boundary.
The unit length of this graph is the mean distance between two individual subscribers.
The two graphs in Fig. \ref{error_wave} look similar.
Therefore, the characteristics shown here seem valid for many cases. (i) If $\lambda \lessapprox 0.5$ or $w< 1$, the strip model is valid. That is, the strip model can cover the boundary bump occurring in a period shorter than half the mean distance between two individual nodes or bumps smaller than half the mean distance between two individual nodes. We can ignore the variation in the boundary smaller than or within an interval shorter than half the mean distance between two individual nodes. (ii) The relative absolute error rapidly becomes large when $w\gtrapprox 1$ and $\lambda\gtrapprox 0.5$. Therefore, if variations in the boundary satisfy $w\gtrapprox 1$ and $\lambda\gtrapprox 0.5$, we need the model to take into account this variation.

\begin{figure}[htb] 
\begin{center} 
\includegraphics[width=8cm,clip]{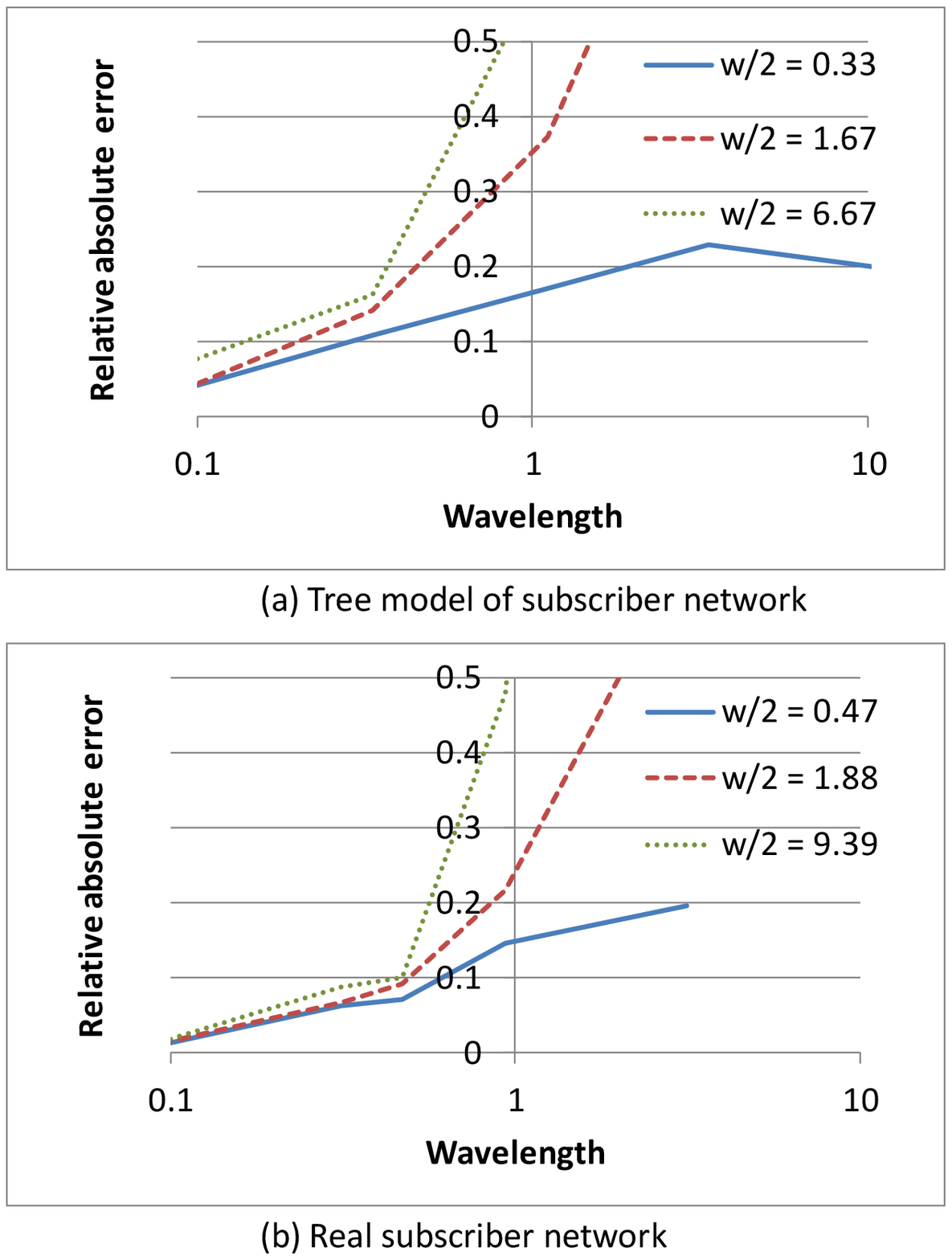} 
\caption{Relative absolute error of $\Pr(i\leftrightarrow l_c \leftrightarrow j)$ under sine wave boundary of disaster}\label{error_wave} 
\end{center} 
\end{figure}

\section{Conclusion}
One geometric model was proposed for evaluating the impact of a disaster on a network and was analyzed through integral geometry (geometric probability).
The validity of this model was evaluated through simulation.  
The simulation results show that the proposed model is valid when the bumps are small or the interval of two bumps are short compared with the mean distance between two nodes.

Performance metrics were derived, such as the probability of maintaining connection between two nodes, as explicit functions of physical route shape.
They are linear functions of perimeter lengths of convex hulls determined by physical route shape.

The results showed the following rules of thumb for designing a network robust against disasters when the disaster area is much larger than the (sub)network of interest:
(1) Reducing the convex hull of the path configuration reduces the expected number of nodes that cannot connect to the destination. 
(2) For two given nodes $n_i,n_j$ on a ring-type network, $\Pr(i\leftrightarrow j)$ is independent of the physical route of the ring-type network. That is, we cannot change $\Pr(i\leftrightarrow j)$ by changing the physical route.
(3) When a single route is provided between two nodes, the straight-line route maximizes the probability maintaining the connectivity between them irrespective of the disaster area.
When it is difficult to adopt the straight-line route, the effect of introducing a loop is identical to that.

It is assumed that a disaster area can be modeled by an area randomly placed around a network.
This assumption is based on the consideration that we cannot forecast the locations or shapes for many types of disasters.
In particular, the relative location of the disaster area and subnetworks is not known.
Therefore, under no such prior information, the proposed model is valid as a first-step approximation.
However, for a specific type of disaster, we may obtain information about the disaster location.
For such cases, an assessment method for network survivability using such information is for further study. 
In addition, the assumption of convexity of the ring-type network and the assumption that no network elements function in $D$ are also first-step approximations.
The results under the relaxed assumptions for these assumptions will be presented in the near future.

\begin{biography}
{Hiroshi Saito} graduated from the University of Tokyo with a B.E. degree in Mathematical Engineering in 1981, an M.E. degree in Control Engineering in 1983, and a Dr.Eng. in Teletraffic Engineering in 1992.

He joined NTT in 1983. He is currently an Executive Research Engineer at NTT Service Integration Labs. He received the Young Engineer Award of the Institute of Electronics, Information and Communication Engineers (IEICE) in 1990, the Telecommunication Advancement Institute Award in 1995 and 2010, and the excellent papers award of the Operations Research Society of Japan (ORSJ) in 1998. He served as an editor and a guest editor of technical journals such as Performance Evaluation, IEEE Journal of Selected Areas in Communications, and IEICE Trans. Communications, the organizing committee chairman of a few international conferences, and a TPC member of more than 30 international conferences. He is currently an editorial board member of Computer Networks, and the director of Journal and Transactions of IEICE. Dr. Saito is a fellow of IEEE, IEICE, and ORSJ, and a member of IFIP WG 7.3. His research interests include traffic technologies of communications systems, network architecture, and ubiquitous systems.

More information can be found at http://www9.plala.or.jp/hslab.
\end{biography}

\end{document}